\documentclass{sig-alternate}

\usepackage{graphicx}
\usepackage{subfigure}

\usepackage{array}
\usepackage{extarrows}

\usepackage{algorithm}
\usepackage{algorithmicx}
\usepackage{algpseudocode}
\usepackage{booktabs}

\usepackage{color}

\begin{document}

\title{
Mining Frequent Neighborhood Patterns in Large Labeled Graphs
\titlenote{The work was done when the first author was visiting Microsoft Research Asia.}
}

\numberofauthors{2}

\author{
\alignauthor
Jialong Han\\
       \affaddr{Renmin University of China}\\
       \email{jialonghan@gmail.com}
\alignauthor
Ji-Rong Wen\\
       \affaddr{Microsoft Research Asia}\\
       \email{jrwen@microsoft.com}
}

\maketitle

\conferenceinfo{KDD'13}{August 11-14, 2013, Chicago, USA.}
\CopyrightYear{2013}

\begin{abstract}

Over the years, frequent subgraphs have been an important sort of targeted patterns in the pattern mining literatures, where most works deal with databases holding a number of graph \emph{transactions}, e.g., chemical structures of compounds.
These methods rely heavily on the \emph{downward-closure property} (DCP) of the support measure to ensure an efficient pruning of the candidate patterns.
When switching to the emerging scenario of single-graph databases such as Google Knowledge Graph and Facebook social graph, the traditional support measure turns out to be trivial (either 0 or 1).
However, to the best of our knowledge, all attempts to redefine a single-graph support resulted in measures that either lose DCP, or are no longer semantically intuitive.

This paper targets mining patterns in the single-graph setting.
We resolve the ``DCP-intuitiveness'' dilemma by shifting the mining target from frequent subgraphs to frequent neighborhoods.
A neighborhood is a specific topological pattern where a vertex is embedded, and the pattern is frequent if it is shared by a large portion (above a given threshold) of vertices.
We show that the new patterns not only maintain DCP, but also have equally significant semantics as subgraph patterns.
Experiments on real-life datasets display the feasibility of our algorithms on relatively large graphs, as well as the capability of mining interesting knowledge that is not discovered in prior works.

\end{abstract}

\category{H.2.8}{Database Management}{Database Applications}[Data mining]

\terms{Algorithms, Experimentation}

\keywords{Graph Mining}

\newtheorem{theorem}{Theorem}
\newtheorem{property}{Property}
\newdef{definition}{Definition}
\newtheorem{example}{Example}

\section{Introduction}

\begin{figure}
  \centering
    \subfigure[47.0\% of authors have at least two papers.]{
    \label{fig:AuthorPaper:b} %% label for second subfigure
    \includegraphics[scale=0.75]{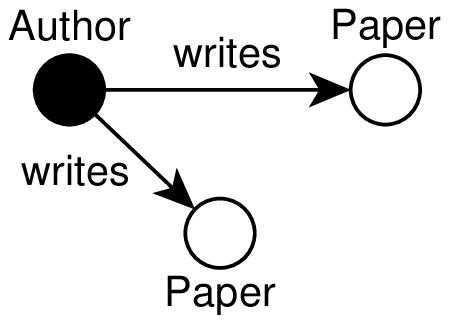}}
  \hspace{0.1\columnwidth}
 \subfigure[9.1\% of authors once cited their own paper.]{
    \label{fig:AuthorPaper:a} %% label for first subfigure
    \includegraphics[scale=0.75]{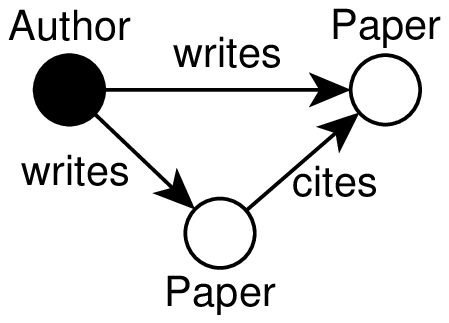}}
  \caption{Neighborhood patterns mined from a public citation network dataset, together with support ratios}
  \label{fig:AuthorPaper} %% label for entire figure
\end{figure}

Since Agrawal et al. introduced the concept of Association Rule Mining\cite{ARM} in 1993, frequent itemset mining, which is the core subtask of association rule mining, has resulted in fruitful follow-up works.
Among the horizontal explorations which target mining substructures more expressive than a subset, including subsequences, subtrees, and subgraphs, the Frequent Subgraph Mining (we denote it as FSM later for short) problem turns out to be the most expressive.

In the typical \emph{graph-transaction setting}, the database consists of a large number of \emph{transactions}, e.g., chemical structures of small molecules.
The measure of frequency, a.k.a. support, is then naturally defined as how many transactions a given pattern is observed as a subgraph of.
This definition provides clear semantics in applications.
For example, an atom group commonly found among a set of organic compounds may indicate that they potentially share some properties.
Moreover, it also satisfies the \emph{downward-closure property} (DCP), which requires that the support of a pattern must not exceed that of its sub-patterns.
This property is essential to all frequent pattern mining algorithms, as it enables safely pruning a branch of infrequent patterns in the search space for efficiency.

Nevertheless, when switching to the \emph{single-graph setting}, i.e., the database is itself a large graph and the knowledge inside the single graph is of major concern, the definition of support by counting transactions easily fails because the support of any pattern is simply 0 or 1.
In other words, this definition cannot quantify our intuition that a subgraph occurs ``frequently'' in a large graph.

Indeed, it is not difficult to obtain a support definition based on the count of ``distinct'' matches.
However DCP simply does not hold for any straightforward ones.
Consider Figure \ref{fig:AuthorPaper:b} describing the event ``author $X$ writes paper $Y$ and $Z$''.
When it is matched to a toy database consisting of exactly \underline{one} author writing $n$ papers, the number of different matches for the three vertices is $n(n-1)$.
For the sub-pattern ``author $X$ writes one paper $Y$'', the count is $n<n(n-1)$.
Even if we take into consideration the automorphism of the pattern, and regard all matches as identical if they involve the same set of vertices in the database, it is still the case that $n< {n\choose 2}$.

Intuitively, complicated patterns have larger support counts because they tend to reuse elements in the database.
Inspired by this observation, Vanetik et al.\cite{MIS2_PathAsBuildingBlock} and Kuramochi et al.\cite{MIS1} redefined the support in the single-graph setting to be the maximum number of edge-disjoint matches, which satisfies DCP.
According to them, the support of the ``an-author-writes-two-papers'' pattern on the toy database should be $\lfloor n/2\rfloor$ ($<n$), since we can find almost such number of matches without reusing any $\xrightarrow{writes}$ edge.
Besides the problem complexity increased, we argue that they introduces non-determinism to the support computation, which disobeys the human sense that counting is a ``one-by-one'' procedure.

\begin{figure}
\centering
\epsfig{file=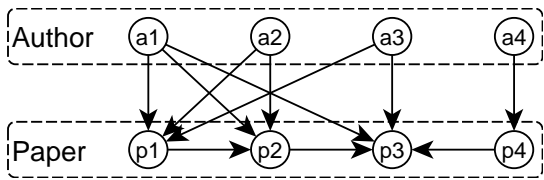,scale=0.7}
\caption{A toy database. ``writes'' and ``cites'' labels are omitted.\label{fig:ToyData}}
\end{figure}

Since it cannot be avoided that traditional matches reuse elements, it seems to be the fact that DCP and intuitiveness can never be both achieved in any subgraph support that counts matches of the entire pattern.
However, if we assign the count operator on a fixed vertex in a pattern, and treat two matches as identical if they match the fixed vertex to the same vertices in the database, we obtain a support measure that both has DCP and is intuitive.
Consider Figure \ref{fig:AuthorPaper:b} again, where the author vertex is painted solid and to be counted.
On the toy database in Figure \ref{fig:ToyData}, though $(a_1,p_1,p_2)$ and $(a_1,p_1,p_3)$ are both matches for $(X,Y,Z)$, they only contribute one to the support because they share the same author $a_1$.
Moreover, $a_2$ and $a_3$ may also serve to compose legal matches, so the overall support is 3.
Similarly, for Figure \ref{fig:AuthorPaper:a} which is a super-pattern of Figure \ref{fig:AuthorPaper:b}, only by matching the author vertex to $a_1$ or $a_2$ can we appropriately arrange the two paper vertices, so the support is 2 (<3).
In fact, under our new definition of support, the two patterns describe ``authors who have at least two papers'' and ``authors who once cited their own paper'', respectively.
Since this sort of pattern characterizes vertices that are embedded in a given local topology, we denote them as \emph{neighborhood patterns}, and the corresponding mining problem as \textbf{F}requent \textbf{N}eighborhood \textbf{M}ining (denoted as FNM for short).
By neighborhood we refer to not only other vertices directly linked to the counted vertex as defined in the graph theory terminology, but also the vertices, edges indirectly connected, along with their labels.

Prior to us, \cite{TreeQuery,GraphProperty} have already studied similar problems by defining the number of such ``partial matches'' as the support of a graph structure.
However, only tree-like patterns were addressed as their mining targets.
Instead, we try to remove the constraint that cycles are not allowed, and investigate the new type of pattern in a generalized way that the FSM problem was studied.
Our contribution lies in that we established rich and deep connections between the two problems from the aspects of basic definitions, problem complexity, solutions, and possible optimizations.
By operating on a real-life dataset we also confirm that trading the problem complexity for better expressivity is worthwhile, for patterns with cycles can lead to more informative and interesting discoveries on the data being investigated.
E.g., taking both Figure \ref{fig:AuthorPaper:a}, \ref{fig:AuthorPaper:b}, and the support ratios in their captions into consideration, we can conclude that among all authors who are ``able'' to cite their own paper (having at least two papers), one out of five will do so.

The rest of this paper is organized as follow.
In Section 2 we formalize the FNM problem, where the \emph{Pivoted Subgraph Isomorphism} problem is identified as the core of FNM, like what subgraph isomorphism is to the FSM problem.
Section 3 discusses our basic solution and further optimization for FNM.
We prove that the building blocks of FNM are not as trivial as those of FSM, while some optimization for the latter one can still be adapted for ours.
In Section 4 we conduct experiments on real datasets to verify the performance of our solution and the utility of the mined neighborhood patterns.
After introducing related and future works we finally conclude.

\section{Problem Formulation}

In this section, first, we introduce basic notations to describe a labeled graph and a neighborhood pattern.
With the notations we then formulate the decision problem of checking whether a neighborhood pattern matches a given vertex in a large graph as the Pivoted Subgraph Isomorphism problem.
We prove that, as the name indicates, this problem is \textsc{np}-complete, making our problem as difficult as the FSM one.
After defining the support of a neighborhood pattern as the number of vertices in the database it could be matched to, we briefly justify its downward closure property.
Finally, more space will be given to some discussions on the expressivity of our problem formulation.

\subsection{Labeled Graphs}

\begin{definition}
A (directed) \textbf{labeled graph} is a 5-tuple $G=\langle V,L_V,E,\Sigma_V,\Sigma_E\rangle$, where
\begin{itemize}
\item $V$ is the set of all vertices;
\item $\Sigma_V$ and $\Sigma_E$ denote label names used to form vertex and edge labels, respectively;
\item $L_V\subseteq V\times \Sigma_V$ is the set of all vertex labels;
\item $E\subseteq V\times V\times \Sigma_E$ is the set of labeled edges;
\end{itemize}
\end{definition}

Note that unlike \cite{AGM,ClosedPattern,FSG}, we allow an arbitrary number of labels on a single vertex.
This is a reasonable generalized assumption for possible applications.
For example in a knowledge base consisting of objects and their relationship, an object may be a father, a politician, and a vegetarian at the same time.
It's also possible that a vertex has no label, i.e., we know nothing about the object, except its existence.
On the other hand, parallel edges carrying distinct label names may link a pair of vertices to model multiple relations simultaneously existing between two objects.
We do not allow edges with no label because we do not process the weak relation of ``arbitrary'' or ``universal'' relation.
Without losing any generality, we do not allow loops, i.e., edges starting and ending with the same vertex.
We can use a vertex label with a specially designed name to replace the loop on a vertex.
We use ``elements'' as the joint name of vertex labels and labeled edges, and define $(|L_V|+|E|)$, i.e., the number of elements, as the size of a labeled graph.

\subsection{Pivoted Subgraph Isomorphism}

\begin{definition}
A \textbf{pivoted graph} is a tuple $\mathcal{G}=\langle G,v_f\rangle$, where
\begin{itemize}
\item $G$ is a labeled graph;
\item $v_f\in V(G)$ is called the \emph{pivot} of $\mathcal{G}$.
\end{itemize}
\end{definition}

Actually, neighborhood patterns are essentially pivoted graphs.
By introducing the concept ``pivot'', we aim to characterize the semantics of fixing a vertex in a subgraph to form a neighborhood pattern, or selecting a vertex in a database to match a pattern to.

\begin{definition}
A pivoted graph $\mathcal{G}_1$ is \textbf{pivoted subgraph isomorphic} to $\mathcal{G}_2$, denoted as $\mathcal{G}_1 \subseteq_f \mathcal{G}_2$, if and only if there exists an injective $f:V(\mathcal{G}_1)\rightarrow V(\mathcal{G}_2)$ such that
\begin{itemize}
\item $\forall (v,l)\in L_V(\mathcal{G}_1)$, $(f(v),l)\in L_V(\mathcal{G}_2)$;
\item $\forall (v_1,v_2,l)\in E(\mathcal{G}_1)$, $(f(v_1),f(v_2),l)\in E(\mathcal{G}_2)$;
\item $f(v_f(\mathcal{G}_1))=v_f(\mathcal{G}_2)$.
\end{itemize}
\end{definition}

The first two descriptions describe that the isomorphic function preserves both vertex labels and edge labels.
In addition, the special isomorphism between pivoted graphs requires that the isomorphic function maps the pivot of $\mathcal{G}_1$ to that of $\mathcal{G}_2$.
As a subtask of FNM, the problem of deciding whether a pivoted graph is pivoted subgraph isomorphic to another is in \textsc{np}-complete.
\begin{theorem}\label{theorem:npc}
The problem of testing pivoted subgraph isomorphism between two arbitrary pivoted graphs is \textsc{np}-complete.
\end{theorem}

We prove in the appendix by reducing it to the classical subgraph isomorphism problem.

\begin{property}\label{property:trans}
The relation $\subseteq_f$ is transitive.
\end{property}

Proof is omited due to the limited space.

\subsection{Support Measure and its PCP}\label{sec:PCP}
\begin{definition}\label{def:support}
Given a large labeled graph $G$, a neighborhood pattern $\mathcal{P}$ \textbf{matches} $v\in V(G)$, if $\mathcal{P} \subseteq_f \langle G,v \rangle$.
Denoting the set of all vertices in $G$ that $\mathcal{P}$ matches as $M_G(\mathcal{P})=\{v\in V(G)| \mathcal{P}\subseteq_f\langle G,v \rangle\}$, we define the support of $\mathcal{P}$ in $G$ as the size of $M_G(\mathcal{P})$, and call $\mathcal{P}$ a \textbf{frequent neighborhood pattern} of $G$, if its support is above a given threshold $\tau$.
\end{definition}

With the support measure defined, the frequent neighborhood mining problem is simply finding all frequent neighborhood patterns in a large graph, with respect to a given threshold.
To control the problem complexity, we further requires the mined patterns be connected, i.e., paths exists between every vertex and the pivot.
In later discussions, sometimes we consider the operation of removing a labeled edge from a pattern.
If the removal leads to an \emph{isolated} vertex, i.e., a vertex without any vertex label or edge associated to it, we further remove the vertex to make the resulted pattern legal.
If we adopt $\subseteq_f$ to describe the sub-pattern/super-pattern relationship between neighborhood patterns, the fact that a pattern cannot be more frequent than any of its sub-patterns is directly derived via Property \ref{property:trans}.

\begin{theorem}
The support measure defined in Definition \ref{def:support} satisfies the downward closure property.
\end{theorem}

\begin{proof}
(Sketch) Given $G$, $\mathcal{P}_1$, and $\mathcal{P}_2$, where $\mathcal{P}_1 \subseteq_f \mathcal{P}_2$.
For any $v\in G$, if $\mathcal{P}_2 \subseteq_f \langle G,v\rangle$, according to Property \ref{property:trans} we have
$\mathcal{P}_1 \subseteq_f \langle G,v\rangle$.
Therefore, $M_G(\mathcal{P}_1)\supseteq M_G(\mathcal{P}_2)$ and it holds that $|M_G(\mathcal{P}_1)|\ge |M_G(\mathcal{P}_2)|$.
\end{proof}

\section{Mining Algorithms}

In this section, we describe the algorithm for mining frequent neighborhood patterns, which follows the apriori breadth-first search paradigm.
We reveal the major technical difference between mining subgraph and neighborhood patterns, that is, the latter task has more complicated ``building blocks''.
We prove that the building blocks in our task are no longer all frequent size-1 patterns.
Instead, they consists of all \emph{frequent paths}, and require special treatments.
Besides, similarities between the solutions and optimizations of FNM and FSM are also discussed.

\subsection{Building Blocks}

Typically, the traditional FSM algorithm generates subgraph patterns in an increasing-size manner.
First, all frequent subgraphs of size 1 are pre-computed as ``building blocks''.
Then candidates of size $K$ are obtained by joining pattern pairs of size $(K-1)$ that differ by only one vertex or edge, after which false positives are filtered by a verification against the database.
We indicate that the join ensures a complete result because every candidate of size $K\ge 2$ is \emph{decomposable}, that is, we can always find \underline{two} distinguished elements, after removing either one we obtain a connected, thus legal, sub-pattern of size $(K-1)$.
They may be isomorphic to each other, but their join takes the candidate into our consideration.

For our neighborhood mining problem, however, it is not the case.
Consider the path-like neighborhood patterns in Figure \ref{fig:PathPattern}.
Obviously, to find a connected sub-pattern of size $(K-1)$ for Figure \ref{fig:PathPattern:nolb}, we have only one choice of removing the edge and vertex at the end of the path.
Meanwhile, in the case of Figure \ref{fig:PathPattern:lb}, only the vertex label to the right can be removed.
Otherwise, the resulted patterns will be illegally unconnected.
Since they are not \emph{decomposable}, it's impossible to derive them by joining two smaller patterns.
Luckily, the following theorem clarifies that these special patterns are only limited to what is described in Figure \ref{fig:PathPattern} and Definition \ref{def:PathPattern}, which enables us to treat them as building blocks and pre-process them in advance.

\begin{figure}
  \centering
  \subfigure[Path pattern without vertex labels]{
    \label{fig:PathPattern:nolb} %% label for first subfigure
    \begin{minipage}[b]{0.5\textwidth}
    \centering
    \includegraphics[scale=0.5]{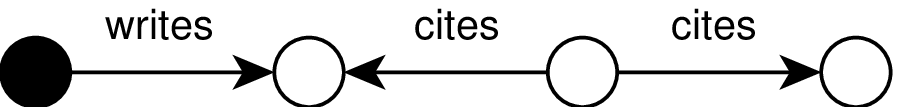}
    \end{minipage}
    }
    \\[10pt]
  %\hspace{0.1\columnwidth}
  \subfigure[Path pattern with a vertex label]{
    \label{fig:PathPattern:lb} %% label for second subfigure
    \begin{minipage}[b]{0.5\textwidth}
    \centering
    \includegraphics[scale=0.5]{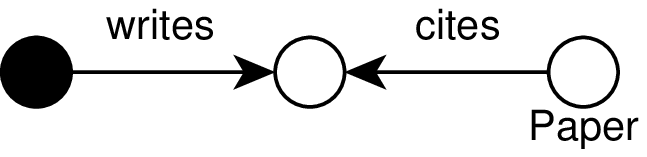}
    \end{minipage}
    }
  \caption{Two variations of path patterns}
  \label{fig:PathPattern} %% label for entire figure
\end{figure}

\begin{definition}\label{def:PathPattern}
A neighborhood pattern is a \textbf{path pattern} if the following statements holds
\begin{itemize}
\item It is a path of labeled edges (directions are ignored) where the pivot is on one end of the path.
\item It contains at most one vertex label, which (if exists) must appear on the other end of the path.
\end{itemize}
\end{definition}

\begin{theorem}\label{theorem:main}
A neighborhood pattern is not decomposable, iff. it is a path pattern.
\end{theorem}

\begin{proof}

\begin{figure}
  \centering
  \subfigure[Case 1]{
    \label{fig:Proof:1}
    \includegraphics[scale=0.4]{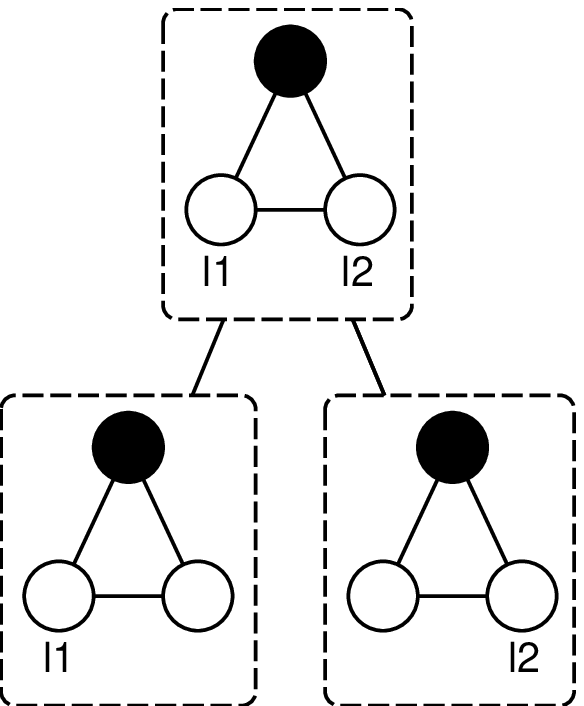}}
  \hspace{0.05\columnwidth}
  \subfigure[Case 2]{
    \label{fig:Proof:2}
    \includegraphics[scale=0.4]{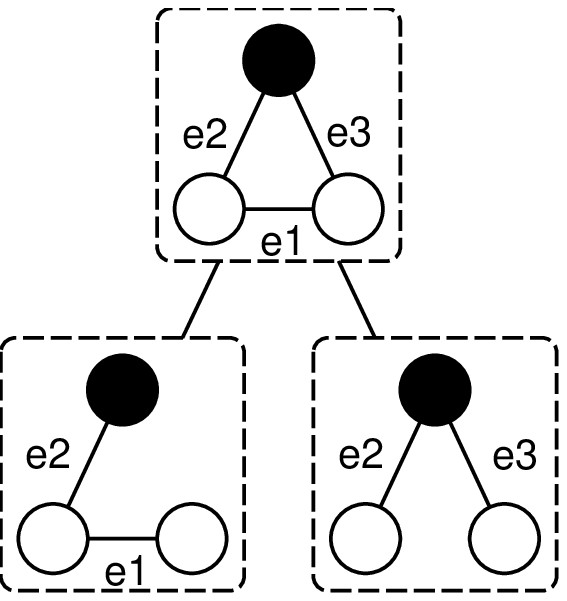}}
    \hspace{0.05\columnwidth}
  \subfigure[Case 3]{
    \label{fig:Proof:3}
    \includegraphics[scale=0.4]{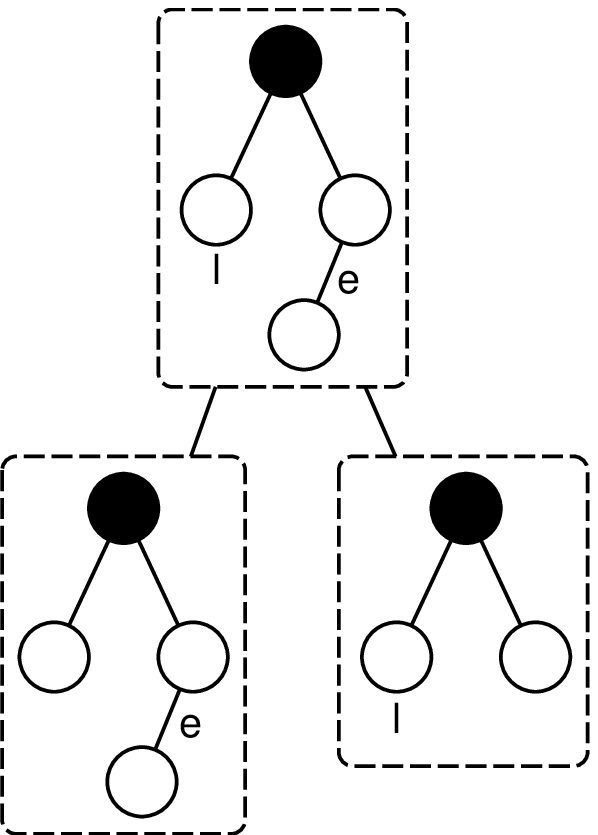}}
    \hspace{0.05\columnwidth}
  \subfigure[Case 4]{
    \label{fig:Proof:4}
    \includegraphics[scale=0.4]{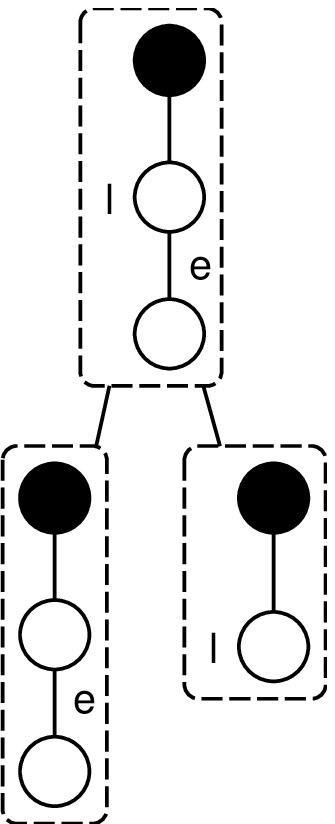}}
  \caption{Decomposable cases in the proof of Theorem \ref{theorem:main}. Labels with no direct influence on the proof are omitted.}
  \label{fig:Proof}
\end{figure}

The sufficiency of the theorem is apparent and has been briefly discussed above.
Therefore we only concentrate on proving the necessity.

If a neighborhood pattern $\mathcal{P}$ is not decomposable, it must have at most one vertex label.
Otherwise, we can arbitrarily choose two of them as $l_1$ and $l_2$, and decompose the pattern as $\mathcal{P}\setminus\{l_1\}$ and $\mathcal{P}\setminus\{l_2\}$, as is illustrated in Figure \ref{fig:Proof:1}.
Moreover, it must not contain cycles.
Otherwise, we arbitrarily choose two edges on the cycle as $e_1$ and $e_2$, and decompose it as $\mathcal{P}\setminus\{e_1\}$ and $\mathcal{P}\setminus\{e_2\}$ (Figure \ref{fig:Proof:2}).
Note that this does not harm the connectivity of the patterns since edges on a cycle are not cutting-edges.

So far, the shape of $\mathcal{P}$ has been limited to be a tree with at most one label.
We transform it to a rooted one, where the root is the pivot of the pattern.
This tree must have only one leaf.
If two, we can again remove them with associated edges respectively (in the case where the leaf possesses the only vertex label, we only remove the label instead) to decompose the pattern (Figure \ref{fig:Proof:3}).

Now the tree with only one leaf is actually a path.
But we still have to prove that if the tree contains a vertex label, it must be on the only leaf: if any vertex other than the leaf carries the label, removing the label and removing the leaf with associated edge respectively will get the pattern decomposed (Figure \ref{fig:Proof:4}).
\end{proof}

\subsection{Constructing Building Blocks}

\begin{algorithm}[t]
\caption{Building block construction\label{algo:BuildingBlock}}
\begin{algorithmic}[1]
\Require
    The single-graph database $G$, minimum support $\tau$
\Ensure
    All frequent path patterns

\State $queue\gets\{\epsilon\}$
\Repeat
\State $path\gets queue.Dequeue()$
\State $count[].Clear()$
\ForAll{$v \in V(G)$}\label{algoLine:vidCount1}
\ForAll{$nextStep \in v.Traverse(path)$}
\State $count[nextStep]\gets count[nextStep]+1$\label{algoLine:vidUpdate1}
\EndFor
\EndFor

\ForAll{$nextStep \in count.Keys()$}

\If{$count[nextStep]\ge \tau$}
\State $newPath\gets path.Append(nextStep)$
\State $R\gets R\cup\{newPath\}$
\If{$nextStep.IsEdgeStep()$}
\State $queue.Enqueue(newPath)$
\EndIf
\EndIf

\EndFor

\Until{$queue.Empty()=true$}\\
\Return $R$
\end{algorithmic}
\end{algorithm}

Being a special case of the general neighborhood patterns, path patterns can still be organized into a level-wise structure, or more exactly, a hierarchical one, which preserves the downward closure property.
The parent of each path pattern is uniquely found, by removing the vertex label or the vertex on the other end of the path than the pivot.
Thus, the level-wise search algorithm on such a structure deserves an ``extending'' approach to generate larger patterns from small ones, rather than the ``joining'' one used for non-building-blocks.

In Algorithm \ref{algo:BuildingBlock}, we describe the basic algorithm for finding frequent paths.
First, a queue used for the bread-first search is initialized with an empty path $\epsilon$.
When extending a path on the front of the queue, we traverse according to the pattern, each time with one vertex in $G$ as the starting point.
Note that for each starting point, we should not visit a vertex more than once.
Each traversal returns all possible moves when we arrive at the ending point(s) and try to take a next step.
E.g., we traverse along a ``$\bullet\xrightarrow{writes}\circ \xleftarrow{cites}\circ$'' path starting from vertex ``Jiawei Han'', and stop at the vertex of paper \cite{MIS1} (it cites \cite{ClosedPattern} of Jiawei Han), then possible next-steps on \cite{MIS1} may be following an ``cites'' edge to another unvisited paper it cites (such as paper \cite{AGM}) to produce Figure \ref{fig:PathPattern:nolb}, or terminate the path with a vertex label ``Paper'' to end up with Figure \ref{fig:PathPattern:lb}.
Each time a new next-step for the current starting point is discovered, it increases its counter by 1.
When all traversals are over, those next-steps with a count of more than $\tau$ are used to extend the path.
After saving all extended paths to the result set, non-terminated paths, i.e., new paths obtained by appending an edge rather than a vertex label such as Figure \ref{fig:PathPattern:nolb}, are added to the queue for further expansion.
The algorithm terminates when all extendable paths in the queue are consumed.

\subsection{Joining and Verifying}

\begin{algorithm}[t]
\caption{Frequent neighborhood mining\label{algo:FNM}}
\begin{algorithmic}[1]
\Require
    The single-graph database $G$, minimum support $\tau$
\Ensure
    All frequent neighborhood patterns

\State $fPaths\gets FrequentPaths(G)$
\State $f_1\gets fPaths.SelectSize(1)$
\State $k\gets 2$
\While{$f_{k-1}.Empty()=false$}
\ForAll{$1\le i\le j \le f_{k-1}.Count()$}\label{algoLine:vidJoin2}
\State $c_k\gets c_k\cup Join(f_{k-1}[i],f_{k-1}[j])$ \label{algoLine:join}
\EndFor
\ForAll{$\mathcal{P} \in c_k$}
\If{$G.CountSupport(\mathcal{P})\ge \tau$}\label{algoLine:vidCountUpdate2}
\State $f_k\gets f_k\cup \{\mathcal{P}\}$\label{algoLine:isomorphismcheck}
\EndIf
\EndFor
\State $f_k\gets f_k\cup fPaths.SelectSize(k)$\label{algoLine:difference}
\State $k \gets k+1$
\EndWhile\\
\Return $\bigcup_{k\ge 1}{f_k}$
\end{algorithmic}
\end{algorithm}

As is stated above, what distinguishes our problem from the traditional ones solved by a ``join-verify-join-\dots'' scheme is the fact that, large path patterns cannot be derived by joining smaller patterns, no matter these smaller ones are paths, trees, or else.
Therefore, the working flow of Algorithm \ref{algo:FNM} differs from other apriori-based subgraph mining algorithm only by Line \ref{algoLine:difference}.
This line adds path patterns of the current size to $F_k$, the frequent non-building-block ones of the same size, to ensure that larger patterns relying on them are not lost due to their absence.

Besides, our join operation at Line \ref{algoLine:join} is also worth an detailed explanation.
Roughly speaking, we determine whether two patterns should be joined via deleting one element from the first, and check whether the remaining structure is pivoted subgraph isomorphic to the second.
Notice that there may be multiple isomorphic mappings so the number of results produced by a single join may be more than one.
For each mapping, the deleted element is mapped to the right position in the second pattern, and inserted to produce a joining result.
If the removed element is a vertex label, the join is relatively easy.
But if it is an edge, the operation is a bit tricky.

\begin{figure}
  \centering
  \subfigure[Deletions may produce unconnected medium results.]{
    \label{fig:Join:unconnect}
    \begin{minipage}[b]{0.5\textwidth}
    \centering
    \includegraphics[scale=0.45]{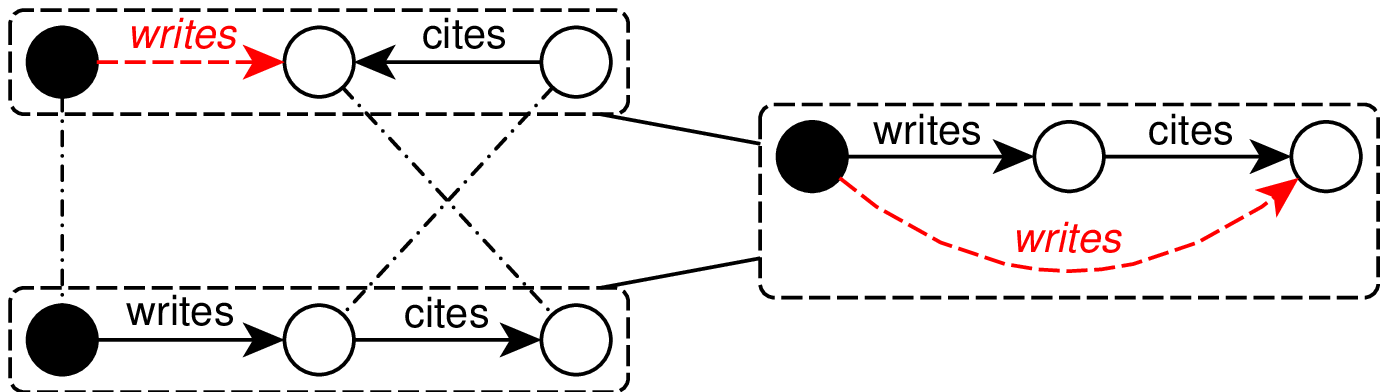}
    \end{minipage}
    }
    \\[10pt]
  %\hspace{0.1\columnwidth}
  \subfigure[Handling dangling edges]{
    \label{fig:Join:dangling} %% label for second subfigure
    \begin{minipage}[b]{0.5\textwidth}
    \centering
    \includegraphics[scale=0.45]{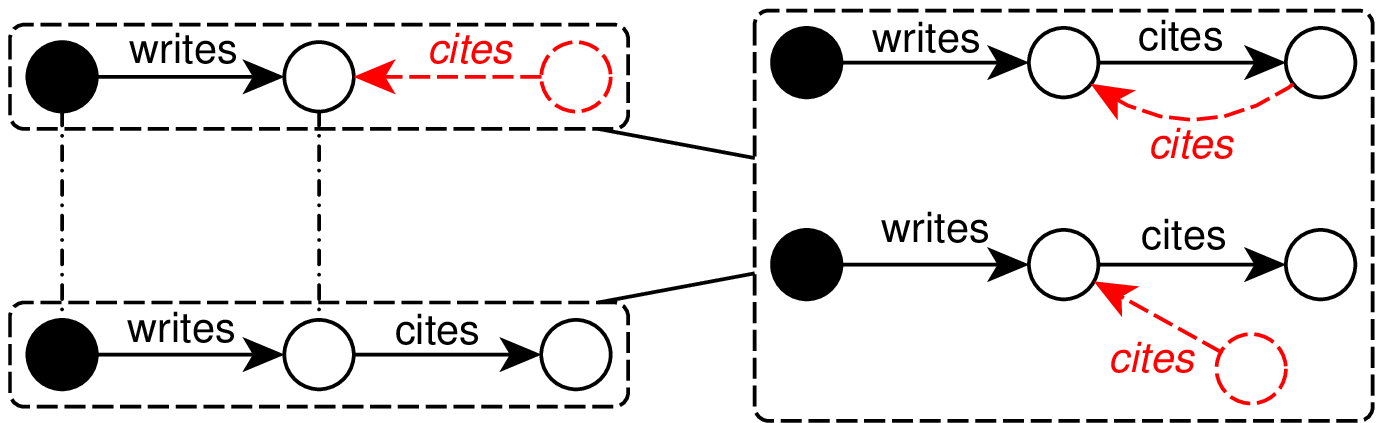}
    \end{minipage}
    }
  \caption{Cases of a single join}
  \label{fig:Join} %% label for entire figure
\end{figure}

On the one hand, the remaining structure is not necessarily connected after the deletion of an edge.
Consider Figure \ref{fig:Join:unconnect}, where we are going to join the patterns ``authors having a cited paper'' and ``authors having a paper citing another''.
For the sake of the example and w.l.o.g., we assume that this join is in a branch of the search space where vertex labels are not introduced yet, while readers can still infer by the context that the pivots are author vertices, and the non-pivot vertices represent papers.
After deleting the $\xrightarrow{writes}$ edge marked with dotted line and italic label in the first pattern we obtain an unconnected structure.
The remaining structure is pivoted subgraph isomorphic to the second pattern, where the mapping is illustrated by dotted lines.
Since the paper vertex in the first pattern that the deleted edge points to is mapped to the second paper vertex in the second pattern, we restore the mapped $\xrightarrow{writes}$ edge between the author vertex and the second paper vertex to generate the result on the right.
Obviously, this pattern is the skeleton of that illustrated in Figure \ref{fig:AuthorPaper:a}.

On the other hand, new vertices may be introduced when handling dangling edges.
In Figure \ref{fig:Join:dangling} we again try joining the same patterns as in Figure \ref{fig:Join:unconnect}, but this time we delete the $\xleftarrow{cites}$ edge.
As is required in Section \ref{sec:PCP}, this deletion isolates the second paper vertex, so the vertex is also removed.
When the pivoted subgraph isomorphism from the remaining structure to the second pattern is established, the vertex that the deleted edge was associated to is mapped to the first paper vertex in the second pattern, which is the new ending point of the restored edge.
But be aware that the new starting point may be the other unmatched paper vertex, as well as a additionally introduced vertex.
Neglecting this case will cause the bottom-right pattern to be lost, which is the representative of all tree-like patterns.

Moreover, it should be noted that Line \ref{algoLine:isomorphismcheck} actually embeds a procedure of checking for duplicated patterns.
If neglected, they will cause more duplicated patterns, joins, and support computations in later computations.
To efficiently check for duplicates, we can hash each produced patterns with the vertex labels on, and associated edges of the pivot.
When a new pattern is produced, we first use the hash table to find potential duplicates, and further verify them with a series of isomorphism checks.

Finally, the last performance overhead of this algorithm lies in the pivoted subgraph isomorphism checker at Line \ref{algoLine:vidCountUpdate2} and \ref{algoLine:isomorphismcheck}.
At this stage, we have not considered adapting any advanced heuristic optimizations of the original subgraph isomorphism problem to ours.
In the experiments we simply implemented a depth-first search checker, utilizing an index built on all label names of the large graph $G$.

\subsection{Optimization via VID-Lists}\label{sec:VID}

In \cite{FSG}, Kuramochi et al. used TID (Transaction Identifier) lists to optimize their FSM algorithm under the graph transaction setting.
Analogously, we propose VID (Vertex Identifier) lists to improve our efficiency both in the building block construction phase and the joining phase.
Both of our optimizations origin from the fact that for any patterns $\mathcal{P}_1\subseteq_f \mathcal{P}_2$, the set of vertices (transactions in Kuramochi's cases) matching $\mathcal{P}_1$, i.e., $M_G(\mathcal{P}_1)$, must be a superset of $M_G(\mathcal{P}_2)$.
This is essentially a reinforced version of the DCP, which enable us to reduce the number of vertices considered when counting the support of a candidate.
To utilize it, we have to maintain the IDs of all vertices in $M_G(P)$ as an ordered list for any $P$, instead of recording only its size.

Specifically, at Line \ref{algoLine:vidCount1} in Algorithm \ref{algo:BuildingBlock}, when extending $path$, we only need to consider $M_G(path)$ instead of all vertices in $G$.
In Algorithm \ref{algo:FNM}, for each enumerated pair of patterns at Line \ref{algoLine:vidJoin2}, we first intersect $M_G(F_{k-1}[i])$ and $M_G(F_{k-1}[j])$ in linear time.
If the number of results is below $\tau$, they need not be joined because vertices matching their shared super-patterns must be within the intersection.
If they pass the test, $M_G(F_{k-1}[i])\cap M_G(F_{k-1}[j])$ is saved, and at Line \ref{algoLine:vidCountUpdate2} we only need to verify the intersection instead of the whole $V(G)$ to count the support of the size-$K$ patterns.
Let's take the join in Figure \ref{fig:Join:unconnect} on the toy database in Figure \ref{fig:ToyData} for example.
The upper left pattern matches author $a_1$,$a_2$, and $a_3$, and the lower left one matches $a_1$,$a_2$,$a_3$, and $a_4$.
Since $a_4$ does not appear in the intersection, it will not be checked when computing the support of the joined pattern.
Vertices matching the pattern under consideration are stored again in VID lists at Line \ref{algoLine:vidUpdate1} of Algorithm \ref{algo:BuildingBlock} and Line \ref{algoLine:vidCountUpdate2} of Algorithm \ref{algo:FNM} for larger patterns' use.

In experiments, the VID-optimization reduces the running time by up to two orders of magnitude. In Section \ref{sec:experiments} we will discuss in detail the experimental results and the feasibility of this optimization.

\section{Experiments}\label{sec:experiments}

Our experiments were performed on two real datasets: EntityCube\footnote{http://entitycube.research.microsoft.com/} and ArnetMiner Citation Network\footnote{http://arnetminer.org/citation} \cite{Tang:07ICDM,Tang:08KDD,Tang:10TKDD,Tang:11ML}, the statistics of which as labeled graphs are presented in Table \ref{tab:dataset}.
Due to the intrinsic characteristics of them, the efficiency of our algorithm was mainly tested on the first one, while the second was used to showcase the distinctive form of knowledge our method is able to discover.
The algorithm was implemented in C\# and run on a 2.4G 16-core Intel Xeon PC with 72GB main memory.
The code optimization option was turned on in the compiler.
All reported time was in seconds.

\begin{table}
\centering
\begin{tabular}{c|ccccc}
\toprule
Dataset&$|V|$&$|L_V|$&$|E|$&$|\Sigma_V|$&$|\Sigma_E|$\\
\midrule
EntityCube&4,685,439&165,533&75,831&288&207\\
ArnetMiner&2,495,972&0&7,791,406&0&3\\
\bottomrule
\end{tabular}
\caption{Two datasets used in the experiments\label{tab:dataset}}
\end{table}

\subsection{Datasets}

The EntityCube system is a research prototype for exploring object-level search technologies, which automatically summarizes the Web for entities (such as people, locations and organizations).
We utilized the relationship network between person entities extracted by the system.
Specifically, with a list of people names as seeds, we queried the system using one name each time, and got related persons and the corresponding relationship names as return.
On the one hand, the seed persons and returned persons were used to form the vertex set $V$.
On the other hand, the system returns two types of relationship.
The name of one is in a plural form, such as ``politicians(Barack Obama, Bill Cliton)'', indicating the connection that they're both politicians.
Thus, the relationship name naturally served as vertex labels for the two associated entities.
The other type of relationship appears in a singular form, e.g., ``wife(Michelle Obama, Barack Obama)''.
They were interpreted as labeled edges between the corresponding vertices.

The ArnetMiner Citation Network dataset contains many papers with associated attribute information, as well as their citation relationship.
The dataset consists of five versions, while we use the fifth one.
We extracted all papers, authors, and conferences appearing in the data, and construct the vertex set with them.
Conferences of the same series and on different years are treated as identical.
There are only three types of labeled edges, i.e., ``writes'' between an author and a paper, ``accepts'' between a conference and a paper, and ``cites'' between a paper and another.
Because these edge label names actually implies the type of both the starting and ending vertices of an edge, we didn't employ any vertex label to avoid redundancy.
In the data, IDs are provided to uniquely denote papers and form citations, which were adopted by us.
However in the author and conference sections of each paper, only texts are presented.
Therefore, when converting them to the IDs in our algorithm, we required exact text-match and didn't perform any cleaning operation involving external data.

\subsection{Performance}

\begin{figure*}[t]
  \centering
  \subfigure[Performance (building block)]{
    \label{fig:EntityCube:bbtime}
    \includegraphics[width=0.65\columnwidth]{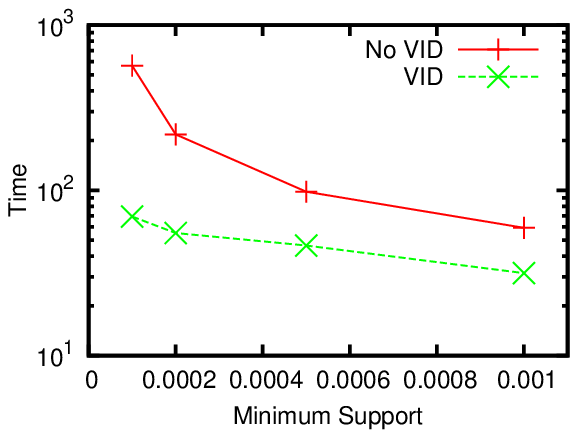}}
    \hspace{0.02\columnwidth}
  \subfigure[Performance (joining \& verifying)]{
    \label{fig:EntityCube:time}
    \includegraphics[width=0.65\columnwidth]{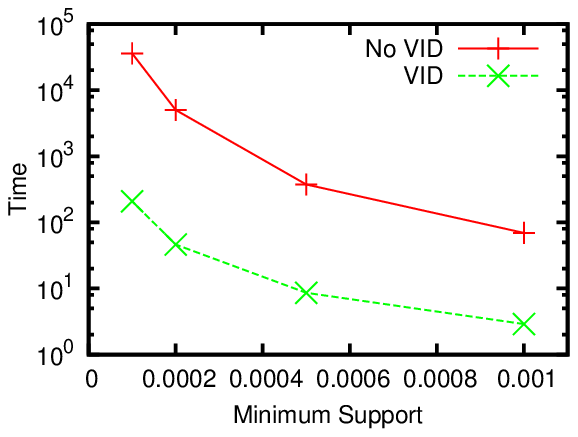}}
    \hspace{0.02\columnwidth}
    \subfigure[Filtering power of VID]{
    \label{fig:EntityCube:cand}
    \includegraphics[width=0.65\columnwidth]{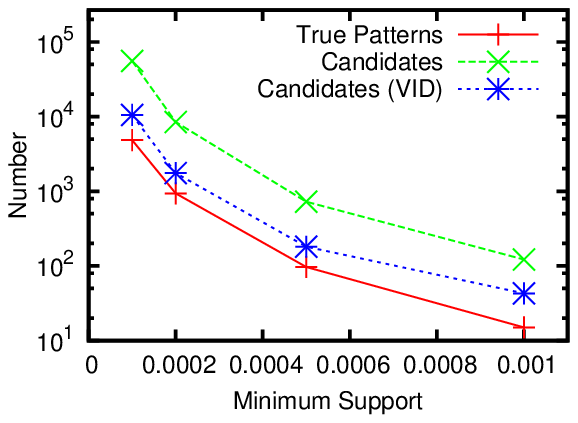}}
  \hspace{0.02\columnwidth}

  \subfigure[Verification time per candidate]{
    \label{fig:EntityCube:perval}
    \includegraphics[width=0.65\columnwidth]{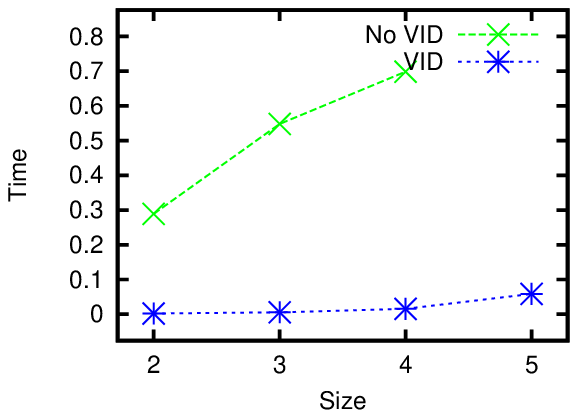}}
    \hspace{0.02\columnwidth}
  \subfigure[Percentage of different pattern types]{
    \label{fig:EntityCube:cycle}
    \includegraphics[width=0.65\columnwidth]{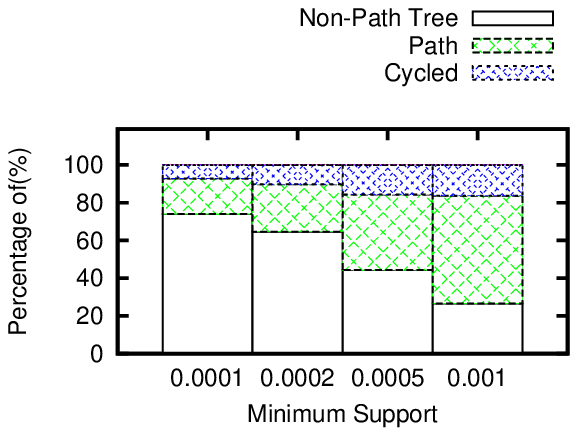}}
  \caption{Experiments on the EntityCube dataset}
  \label{fig:EntityCube} %% label for entire figure
\end{figure*}

In this section, we report the performance of our algorithms for frequent neighborhood mining.
Since no previous work has addressed exactly the same problem, our experiments were dedicated to validating the feasibility of our VID optimization.
In practice, the running time of such a pattern mining algorithm is heavily influenced by the size of the result set.
Therefore, we decided to conduct the experiments on the EntityCube data with rich label names, for it has a potentially larger result set and the running time is more sensitive to parameters such as the minimum support $\tau$.
In all experiments, $\tau$ was chosen from $\{0.0001,0.0002,0.0005,0.001\}$ and all reported time was an average of five consecutive runs.

In Figure \ref{fig:EntityCube:bbtime} and \ref{fig:EntityCube:time}, we terminated the search after all frequent patterns of size below 4 were discovered.
It is clear that our VID optimization successfully accelerates both the building block construction and the join-verify phases by up to one and two orders of magnitude, respectively.
Analogous to the TID optimization \cite{FSG} for FSM, the advantages of the VID optimization are two-fold.
First, two patterns will not be joined if the intersection of their VID lists is smaller than $\tau$.
Therefore, the algorithm successfully avoids verifying false positives caused by joining unpromising pattern pairs, which is a vital overhead to the overall performance.
In Figure \ref{fig:EntityCube:cand}, the number of candidates with/without the VID-list-pruning, and the number of true patterns are illustrated.
This figure shows that the pruning helps narrow down the number of candidates by several times.
Second, for each pair of patterns that passes the pruning, the time spent on counting the support of the joining results is also reduced because vertices not in the intersection won't contribute to the count, thus are not checked.
Figure \ref{fig:EntityCube:perval} presents the average verification time for each non-path candidate where $\tau=0.0001$.
The x-axis denotes different stages of the search procedure, where candidates of size 2 to 5 are verified.
Obviously the VID optimization is significantly effective for verifying candidates of all sizes.
Particularly, without the optimization, the algorithm didn't finish the verification of size 5 in reasonable time.

\subsection{Interpretation of Mined Patterns}

%\begin{table*}
%\centering
%\begin{tabular}{|c|c|b{4cm}<{\centering}|c|}
%\hline
%ID&Pattern&Explanation&Support\\ \hline
%1&\includegraphics[scale=0.5]{94.6.eps}&Percentage of authors that have a co-author.&94.6\%\\ \hline
%2&\includegraphics[scale=0.5]{89.0.eps}&Percentage of conferences that once accepted two papers from the same author.&89.0\%\\ \hline
%3&\includegraphics[scale=0.5]{78.9.eps}&Percentage of authors that have a paper with two co-authors.&78.9\%\\ \hline
%4&\includegraphics[scale=0.5]{37.9.eps}&Percentage of authors that have a co-author cooperating twice with him.&37.9\%\\ \hline
%5&\includegraphics[scale=0.5]{31.4.eps}&Percentage of authors that have at least three papers.&31.4\%\\ \hline
%6&\includegraphics[scale=0.5]{27.6.eps}&Percentage of authors that have a cited paper.&27.6\%\\ \hline
%7&\includegraphics[scale=0.5]{25.4.eps}&Percentage of authors that have two paper cited by one conference.&25.4\%\\ \hline
%8&\includegraphics[scale=0.5]{23.0.eps}&Percentage of conferences that once accepted a paper and another paper this one cited.&23.0\%\\ \hline
%9&\includegraphics[scale=0.5]{19.0.eps}&Percentage of authors that have a paper that was cited at least twice.&19.0\%\\ \hline
%10&\includegraphics[scale=0.5]{9.1.eps}&Percentage of conferences that once accepted two papers citing each other.&9.1\%\\ \hline
%\end{tabular}
%\caption{Selected Interesting Neighborhood Patterns}
%\end{table*}

\begin{figure*}[t]
  \centering
  \subfigure[94.6\% of authors have a co-author.]{
    \label{fig:ArnetMinerPatterns:94.6} %% label for first subfigure
    \begin{minipage}[b]{0.3\textwidth}
    \centering
    \includegraphics[scale=0.5]{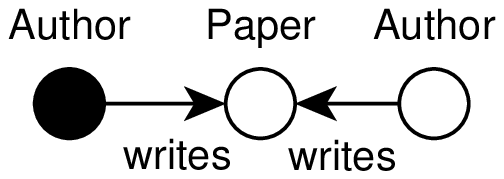}
    \end{minipage}
    }
    \hspace{0.02\columnwidth}
  \subfigure[89.0\% of conferences once accepted two papers from the same author.]{
    \label{fig:ArnetMinerPatterns:89.0} %% label for second subfigure
    \begin{minipage}[b]{0.3\textwidth}
    \centering
    \includegraphics[scale=0.5]{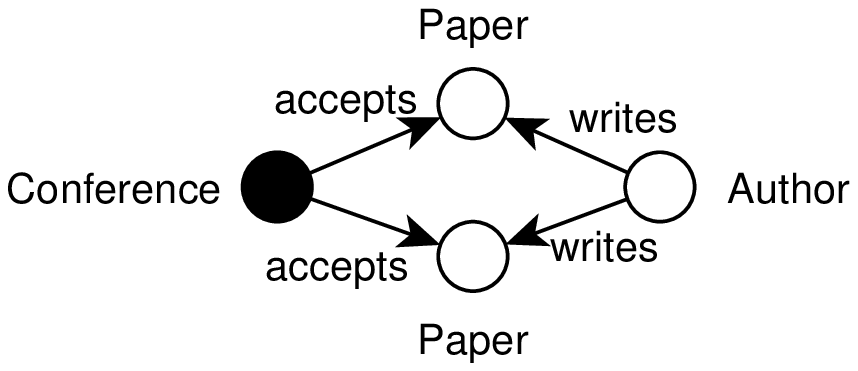}
    \end{minipage}
    }
    \hspace{0.02\columnwidth}
  \subfigure[78.9\% of authors have a paper with two co-authors.]{
    \label{fig:ArnetMinerPatterns:78.9} %% label for second subfigure
    \begin{minipage}[b]{0.3\textwidth}
    \centering
    \includegraphics[scale=0.5]{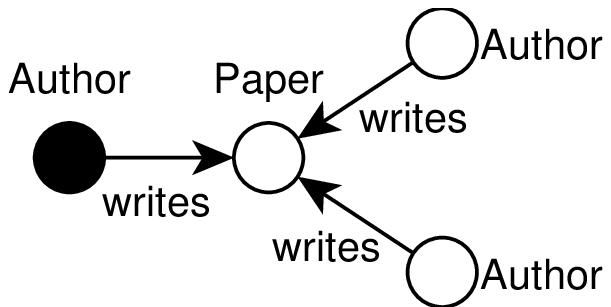}
    \end{minipage}
    }
        \hspace{0.02\columnwidth}
  \subfigure[For about 61.8\% of papers, two of their authors have other cooperations.]{
    \label{fig:ArnetMinerPatterns:61.8} %% label for second subfigure
    \begin{minipage}[b]{0.3\textwidth}
    \centering
    \includegraphics[scale=0.5]{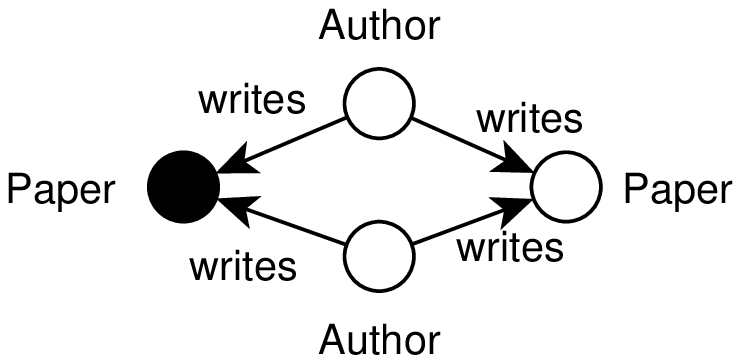}
    \end{minipage}
    }
        \hspace{0.02\columnwidth}
  \subfigure[37.9\% of authors have a co-author cooperating twice with him.]{
    \label{fig:ArnetMinerPatterns:37.9} %% label for second subfigure
    \begin{minipage}[b]{0.3\textwidth}
    \centering
    \includegraphics[scale=0.5]{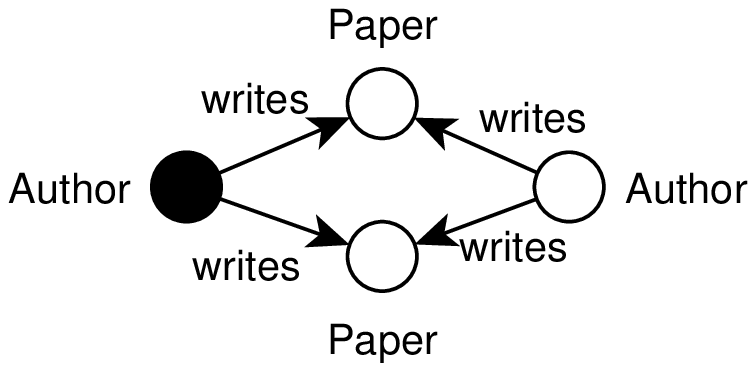}
    \end{minipage}
    }
        \hspace{0.02\columnwidth}
  \subfigure[31.4\% of authors have at least three papers.]{
    \label{fig:ArnetMinerPatterns:31.4} %% label for second subfigure
    \begin{minipage}[b]{0.3\textwidth}
    \centering
    \includegraphics[scale=0.5]{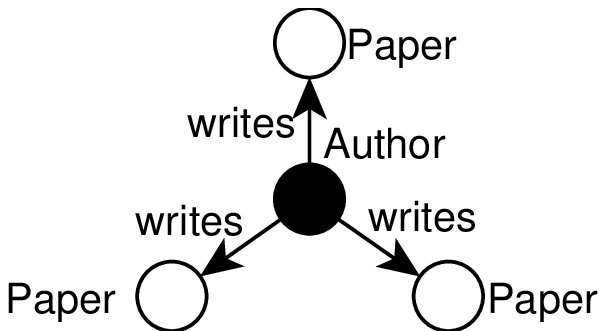}
    \end{minipage}
    }
        \hspace{0.02\columnwidth}
  \subfigure[27.6\% of authors have a cited paper.]{
    \label{fig:ArnetMinerPatterns:27.6} %% label for second subfigure
    \begin{minipage}[b]{0.3\textwidth}
    \centering
    \includegraphics[scale=0.5]{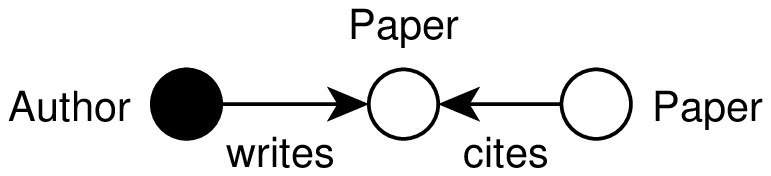}
    \end{minipage}
    }
        \hspace{0.02\columnwidth}
  \subfigure[25.4\% of authors have two papers accepted by one conference.]{
    \label{fig:ArnetMinerPatterns:25.4} %% label for second subfigure
    \begin{minipage}[b]{0.3\textwidth}
    \centering
    \includegraphics[scale=0.5]{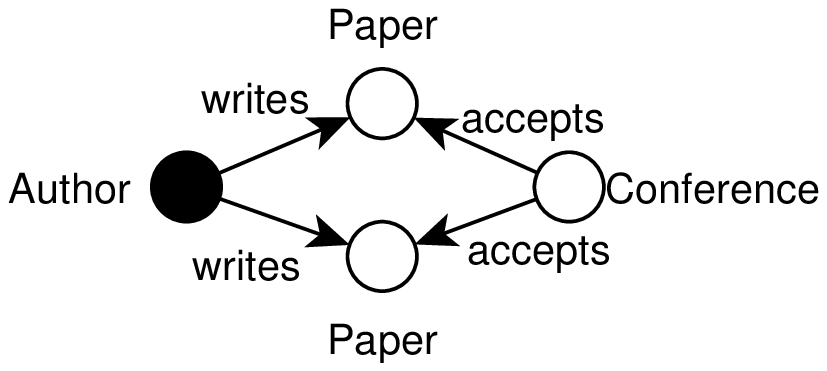}
    \end{minipage}
    }
        \hspace{0.02\columnwidth}
  \subfigure[23.0\% of conferences once accepted a paper and another paper this one cited.]{
    \label{fig:ArnetMinerPatterns:23.0} %% label for second subfigure
    \begin{minipage}[b]{0.3\textwidth}
    \centering
    \includegraphics[scale=0.5]{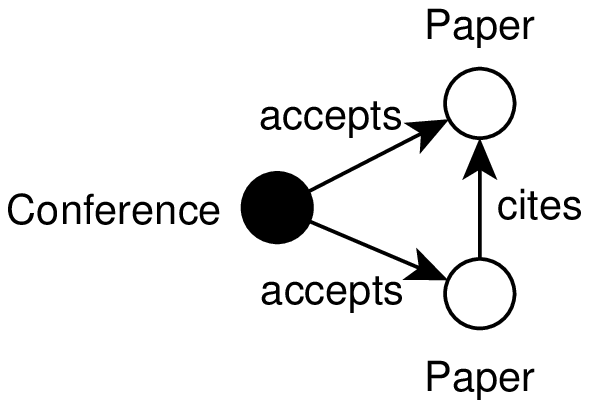}
    \end{minipage}
    }
        \hspace{0.02\columnwidth}
  \subfigure[19.0\% of authors have a paper that was cited at least twice.]{
    \label{fig:ArnetMinerPatterns:19.0} %% label for second subfigure
    \begin{minipage}[b]{0.3\textwidth}
    \centering
    \includegraphics[scale=0.5]{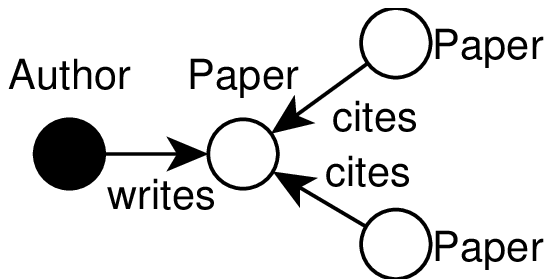}
    \end{minipage}
    }
\hspace{0.02\columnwidth}
  \subfigure[About 9.8\% of papers cite another paper of one of its authors.]{
    \label{fig:ArnetMinerPatterns:9.8} %% label for second subfigure
    \begin{minipage}[b]{0.3\textwidth}
    \centering
    \includegraphics[scale=0.5]{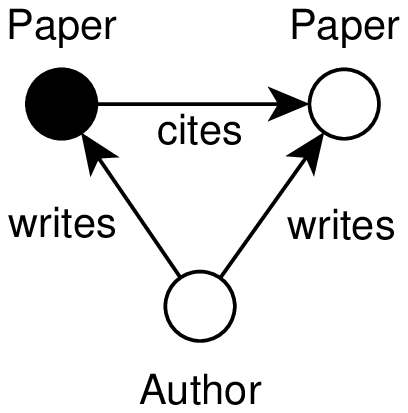}
    \end{minipage}
    }
    \hspace{0.02\columnwidth}
  \subfigure[6.3\% of conferences once accepted two papers citing each other.]{
    \label{fig:ArnetMinerPatterns:6.3} %% label for second subfigure
    \begin{minipage}[b]{0.3\textwidth}
    \centering
    \includegraphics[scale=0.5]{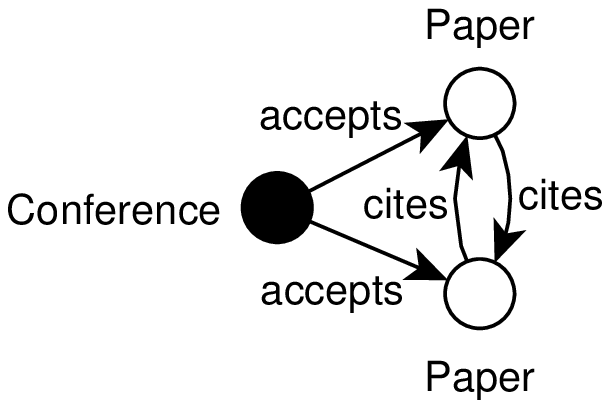}
    \end{minipage}
    }

  \caption{Selected interesting neighborhood patterns in the ArnetMiner dataset.}
  \label{fig:ArnetMinerPatterns} %% label for entire figure
\end{figure*}

As is mentioned above, the major superiority of FNM over other graph pattern mining methods is that it discovers patterns with cycles, which is not targeted by others.
With the hands-on experiences of experimenting on both datasets, we realize that a cycled pattern can be viewed as a set of constraints with lower degree of freedom than a tree-like one of the same size.
Figure \ref{fig:EntityCube:cycle} shows the constitution of all patterns in EntityCube data, whose size are below 5.
Patterns with cycles actually make up around 10\% among all three types.
The trend also shows that when we decrease the support ratio and specify a pattern into its super-patterns, it is more difficult for a cycle to form, than a fork to appear.

However, once formed, patterns with cycles serve as a good complement to tree-like patterns.
Introducing them does not linearly increase our knowledge about the data being investigated, but actually makes a mutual reinforcement with tree-like ones.
As patterns from the ArnetMiner dataset have better interpretability, we selected some interesting neighborhood patterns mined from this dataset to demonstrate our points.
Besides the example given in Figure \ref{fig:AuthorPaper}, more patterns are displayed in Figure \ref{fig:ArnetMinerPatterns}.
By combining two or more of them, we can make very interesting discoveries about the academia.

For example, the support ratio of Figure \ref{fig:ArnetMinerPatterns:27.6} is lower than that of Figure \ref{fig:AuthorPaper:b} and \ref{fig:ArnetMinerPatterns:31.4}, which reflects the common sense that it is more difficult to get one's paper cited than to write more papers.
Besides, the small gap between the ratios of Figure \ref{fig:AuthorPaper:b} and \ref{fig:ArnetMinerPatterns:37.9} reveals the fact that most writers are willing to maintain a co-authoring relationship.
On the other hand, the ratios of Figure \ref{fig:AuthorPaper:b} and \ref{fig:ArnetMinerPatterns:25.4} together prove that an average author relatively favors a conference which once accepted his paper.
Moreover, Figure \ref{fig:ArnetMinerPatterns:78.9} alone points out that most of us (assuming that we all have papers) have a paper with no less than three authors.
Surprisingly, as Figure \ref{fig:ArnetMinerPatterns:6.3} indicates, there are even papers citing each other!
By checking the data we find two cases of such a phenomenon.
One is caused by the dataset itself.
The data treats books as conferences, and their chapters as papers.
Of course, chapters from the same book can cite each other.
This case is rare.
The other case is more often: an author simultaneously submitted two papers to the same conference and got them both accepted.
When preparing the camera-ready versions, he made them citing each other.

\begin{table}
\centering
\begin{tabular}{c|cccc}
\toprule
Vertex Type&Number&Time&Patterns&Cycled\\
\midrule
Conference&6,713&3,354&135&30\\
Author&916,979&9,469&163&24\\
Paper (sampled)&500,000&53,422&796&147\\
\bottomrule
\end{tabular}
\caption{Statistics of the runs on ArnetMiner\label{tab:ArnetMiner}}
\end{table}

For all patterns presented in this paper, we use support ratios w.r.t. vertices with a specified label, instead of the absolute count mentioned in the problem statement.
It's easy to implement, which only involves a small modification to the algorithm.
Suppose we want facts about all authors with a minimum support ratio of 10\%.
In Algorithm 1 and 2, when a scan on the entire $V(G)$ is required for support counting, we only scan those author vertices by accessing an index on the label names.
Moreover, when calling the modified algorithms, $\tau$ should be assigned with 10\% the number of all author vertices.
When mining patterns about authors, conferences, and papers, the support ratios were set to 1\%.
As the number of paper vertices is huge (over 1.5 million), the support calculation was performed on a subset of 0.5 million papers we sampled.
We didn't explore path patterns of size 4, or any pattern whose size exceeds 4.
Readers may refer to Table \ref{tab:ArnetMiner} for more running details.

\section{Related Work}

\subsection{Frequent Subgraph Mining}

The frequent subgraph mining problem is well-investigated by the literatures under the graph-transaction setting.
Among them, the most influential methods are AGM\cite{AGM}, FSG\cite{FSG}, and gSpan\cite{gSpan}.
The first and second adopts the apriori-based BFS scheme, and feature the vertex-incremental and edge-incremental approaches, respectively.
The last one falls into a pattern-growth-based DFS category.
Optimizations they utilized, such as canonical labeling and vertex invariants, are inspiring and potentially employable to our work.
The single-graph setting, however, is not so fully explored due to reasons we mentioned above.
Among the few support measures proposed, the Maximum Independent Set support and corresponding mining algorithms are studied in \cite{MIS1,MIS2_PathAsBuildingBlock}.
\cite{MinimumImageSupport} also defines a single-graph support called Minimum Image Support, which still doesn't make a intuitive one and has yet to be tested for easiness of handling.

\subsection{Frequent Tree Query Mining in Graphs}

In \cite{TreeQuery,TreeQuery2,GraphProperty}, the authors attempt to mine tree patterns in graphs, whose support measure resembles ours in the way that distinct matches of some vertices are counted, while the match conditions on the others are only existential.
They do not target patterns with cycles, which add much to the users' understand of the data.
Moreover, because \cite{TreeQuery,TreeQuery2} allow multiple vertices to be counted (in other words, as our ``pivots''), their problems are more complicated and thus do not completely follow and benefit from the well-solved apriori pattern mining scheme.
We argue that, patterns with more than three pivots may explode in the number, while bringing about some knowledge that is hard to explain and utilize.
Finally, these works both claimed that their mining algorithms support constants in the patterns, e.g., ``x\% of the authors once cited a paper published in KDD''.
Our problem setting supports multiple labels on a vertex.
Therefore, we can achieve it by simply adding the name of each vertex to its label set.
We can also modify our algorithm to implicitly perform such a data transformation.

\subsection{ILP Related Works}

In \cite{WARMR}, Dehaspe et al. introduced an inductive logic programming system for mining frequent patterns in a datalog database.
Their final products are rules, which are more advanced than ours.
However, they require \emph{language biases} as additional inputs to bound the search space.
In contrast, our method is completely unsupervised.
Methods learning horn clauses from knowledge bases such as \cite{LearnHornClause,RandomWorkLaoNi,AMIE} also resemble the Inductive Logic Programming category.
These works are characterized by a variety of metrics to evaluate the utility of a rule.
Since noise and scalability issues in real data are their main concerns, they adopt stricter language biases and the rules mined are of more limited forms.

\section{Future Work}

Under the current problem setting and solution, encouraging results have been achieved in terms of performance and result utility.
However, our work can still be further extended from the following aspects.
First, the definition of \emph{closed} neighborhood patterns may be introduced in a similar way as \cite{ClosedPattern}.
A pattern is closed, if there exists no proper super-pattern with the same support as it.
This definition is expected to significantly reduce the size of results, while preserving the most meaningful ones.
Second, the pivots may be allowed to be an edge to enable characterizing the ``neighborhood'' of an edge.
This generalized pattern introduces new semantics, e.g., ``x\% of all citations are made between papers from the same institutes.''
Third, according to \cite{gSpan}, the depth-first search approach outperforms the apriori-based (breadth-first) approach by an order of magnitude in the FSM task.
Therefore, it is interesting to explore its feasibility in our neighborhood mining task, which has been proved by us to share much in common with the subgraph mining one.

\section{Conclusion}

In this paper, we addressed mining single-graph databases and introduced the new neighborhood patterns as mining targets.
They have clear semantics and are not limited to tree-like shapes.
We formally defined the frequent neighborhood mining problem, and proved that it is as difficult as the frequent subgraph mining problem.
We indicated that the major difference between FNM and FSM in terms of solution is that our patterns have non-trivial building blocks, which are clearly separated by us via a theorem and proof.
After discussing possible optimizations, we conducted experiments on two real datasets to validate the efficiency and effectiveness of our method.
The algorithm is proved to be feasible, and shows an unique ability to provide users with especially interesting insights on the analyzed data.

%\balancecolumns
\bibliographystyle{abbrv}
\bibliography{fnm}

\appendix

\emph{Proof of Theorem \ref{theorem:npc}}.
We prove it by reducing from the subgraph isomorphism problem.
Labels are ignored because it is a generalization of, thus reducible from, the non-label case.

Given an instance $\langle G_1,G_2\rangle$ of the subgraph isomorphism problem, we add a new vertex $v_1$ to $G_1$, and $v_2$ to $G_2$, respectively.
They are marked as pivots and edges are created from them to all vertices of the same graph.
Obviously, $G_1\subseteq G_2$ iff. $\langle G_1,v_1\rangle\subseteq_f \langle G_2,v_2\rangle$.
By solving the pivoted subgraph isomorphism problem $\langle\langle G_1,v_1\rangle,\langle G_2,v_2\rangle\rangle$ we are able to answer whether $G_1\subseteq G_2$.
So our problem is \textsc{np}-hard.
The solution of an instance of our problem is verified in polynomial time.
Therefore, our problem is \textsc{np}-complete.

\end{document}